\documentclass[12pt,onecolumn]{IEEEtran}
\usepackage[ruled,linesnumbered]{algorithm2e}
\usepackage{algorithmic}
\SetKwInput{KwInput}{Input}                
\SetKwInput{KwOutput}{Output}              
\newcommand{\nonl}{\renewcommand{\nl}{\let\nl\oldnl}}

\usepackage{enumerate}
\usepackage{multirow,array}
\usepackage{cite}
\usepackage{graphicx}
\usepackage{psfrag}
\usepackage{url}
\usepackage{amsmath}
\usepackage{amsthm}
\usepackage{array}
\usepackage{algorithm2e}
\usepackage{amssymb}
\usepackage{amsfonts}
\usepackage{float}
\usepackage{textcomp}
\usepackage{xcolor}
\usepackage{tabu}
\usepackage{array}

\usepackage{subfiles}

\newcolumntype{P}[1]{>{\centering\arraybackslash}p{#1}}
\newcolumntype{M}[1]{>{\centering\arraybackslash}m{#1}}

\usepackage{float}

\newtheorem{proposition}{Proposition}
\newtheorem*{conjecture*}{Conjecture}
\newtheorem{lemma}{Lemma}
\newtheorem{observation}{Observation}
\newtheorem{corollary}{Corollary}
\newtheorem{remark}{Remark}
\newtheorem{definition}{Definition}
\newtheorem{theorem}{Theorem}
\newtheorem{example}{Example}

\newcommand{\mc}{\mathcal}
\newcommand{\D}{\Delta}
\newcommand{\de}{\delta}
\newcommand{\ch}{\mathcal{H}}
\title{Bounding the Optimal Length of Pliable Index Coding via a Hypergraph-based Approach}
\begin{document}
\author{
  \IEEEauthorblockN{Tulasi Sowjanya B., Visvesh Subramanian, Prasad Krishnan}
  {
                    }
\vspace{-0.2cm}
}
\maketitle
\begin{abstract}
In pliable index coding (PICOD), a number of clients are connected via a noise-free broadcast channel to a server which has a list of messages. Each client has a unique subset of messages at the server as side-information and requests for any one message not in the side-information. A PICOD scheme of length $\ell$ is a set of $\ell$ encoded transmissions broadcast from the server such that all clients are satisfied. Finding the optimal (minimum) length of PICOD and designing PICOD schemes that have small length are the fundamental questions in PICOD. In this paper, we use a hypergraph-based approach to derive new achievability and converse results for PICOD. We present an algorithm which gives an achievable scheme for PICOD with length at most $\Delta(\mathcal{H})$, where $\Delta(\mathcal{H})$ is the maximum degree of any vertex in a hypergraph that represents the PICOD problem. We also give a lower bound for the optimal PICOD length using a new structural parameter associated with the PICOD hypergraph called the nesting number. We extend some of our results to the PICOD problem where each client demands $t$ messages, rather than just one. Finally, we identify a class of problems for which our converse is tight, and also characterize the optimal PICOD lengths of problems with $\Delta(\mathcal{H})\in\{1,2,3\}$. 
\end{abstract}
\let\thefootnote\relax\footnotetext{
 Visvesh, Tulasi and Dr.\ Krishnan are with the Signal Processing \& Communications Research Center, International Institute of Information Technology Hyderabad, India, email:\{tulasi.sowjanya@research., visvesh.subramanian@research., prasad.krishnan@\}iiit.ac.in. 
 }
\section{Introduction}
The problem of \textit{pliable index coding (PICOD)}, introduced in \cite{brahma2015pliable} by Brahma and Fragouli,  consists of a server with messages, and a number of clients, connected via a noiseless broadcast channel. Each client possesses a subset of messages as \textit{side-information} and demands for \textit{any} message from the set of messages in its \emph{request-set} (those messages not present as its side-information). The server then designs (possibly coded) transmissions and broadcasts them to the clients, which then decode their desired symbols. This is called as a pliable index coding scheme, or a \textit{PICOD scheme}. The number of transmissions made by the server is termed as the \textit{length} of the PICOD scheme, and the goal is to design PICOD schemes with lengths as small as possible. The PICOD problem is a variant of the well-studied index coding problem, introduced in \cite{ISCOD_BiK} by Birk and Kol. The setting in index coding is different from PICOD only in the sense that each client demands a specific message in the request-set, rather than any message. Apart from being an important problem in information theory, pliable index coding was also shown to be useful for constructing data-shuffling schemes for distributed computing in \cite{song2019pliable_shuffling}.

In \cite{brahma2015pliable}, it was shown that finding the optimal code length for PICOD problem is NP-hard. However, via a probabilistic argument, it was shown in \cite{brahma2015pliable} that an achievable scheme with length $\mc{O}(\min(m,n,\log m (1+\log^+(\frac{n}{\log m})))$ exists, where $m$ is the number of messages at the server and $n$ is the number of clients in the system. This means that  $\mc{O}(\log^2n)$ transmissions are sufficient to solve PICOD, when $m=n^\delta$ for some constant $\delta$. A further generalization of PICOD in which each client demands $t$ messages for some $t\geq 1$ was considered in \cite{multiple_req_brahm_frag}. We call such PICOD problems as PICOD($t$) problems. For the special class of PICOD($t$) problems with equal-sized request sets, the work \cite{multiple_req_brahm_frag} presented bounds on the optimal PICOD length. 

Algorithms for PICOD were presented in \cite{brahma2015pliable}, using a greedy approach, via randomized covering, and by utilizing connections to index coding. In \cite{PolyTime_PICOD}, Song and Fragouli presented a deterministic algorithm for PICOD via another greedy technique, which constructed a PICOD scheme with at most $\mc{O}(\log^2n)$ transmissions. This remains the state of the art, with respect to deterministic algorithms for achievability in general PICOD problems. The work \cite{PolyTime_PICOD} also considered the PICOD($t$) scenario, and showed another variant of the same algorithm constructing a PICOD scheme of length $\mc{O}(t\log n + \log^2n)$ transmissions are sufficient.  Achievability schemes and information theoretic converses for some special classes of PICOD problems were derived in \cite{liu2019tight,8849812_ShanuRaj_CodeConstrPIC}. Converse techniques were further improved in \cite{ong2019optimal,ong2019improved} and lower bounds on the optimal PICOD length were presented in the same for general PICOD problems based on their structure. 

In \cite{krishnan2021pliable}, Krishnan et al. considered a hypergraph coloring approach to the PICOD problem. A given PICOD problem can be represented by a \textit{PICOD hypergraph} $\mc{H}$ with the messages as its vertices, and the request-sets of clients as its edges. Using \textit{conflict-free colorings} \cite{even2003conflict} of the PICOD hypergraph, it was shown in \cite{krishnan2021pliable} that there exists a PICOD scheme which achieve length $\mc{O}(\log^2\Gamma)$, where $\Gamma$ is the maximum number of edges intersecting with any edge. This bound is in general tighter than the one derived in \cite{brahma2015pliable,PolyTime_PICOD}, as $\Gamma$ can be much smaller than $n$. However, the proof is via a probabilistic argument again, and a deterministic algorithm for constructing a PICOD scheme of length $\mc{O}(\log^2\Gamma)$ is not known so far.

In the present work, we consider the hypergraph model for PICOD introduced in \cite{krishnan2021pliable} and present new achievability and converse results for PICOD based on this model. After briefly recollecting the formal PICOD system model and the hypergraph framework in Section \ref{sec:systemmodel_and_hypergraph}, in Section \ref{sec:poly-time-delta-algorithm} we show that the optimal PICOD length $\beta(\ch)$ for a PICOD problem represented by hypergraph $\ch$ satisfies $\beta(\ch)\leq \D(\ch)$, where $\D(\ch)$ refers to maximum degree of any vertex in the hypergraph (i.e., the maximum number of request-sets that any message is present in). This bound is better than the $\mc{O}(\log^2n)$ bound in general, for PICOD hypergraphs in which $\D(\ch)$ is smaller than $\log^2 n$. 
As a constructive proof of this result, we present Algorithm \ref{Alg1} in Section \ref{subsec:algorithm}, which uses the greedy approach to achieve this bound, and runs in time polynomial in the system parameters. In Section \ref{sec:lowerbound_via_nesting}, we present a lower bound for $\beta(\ch)$ based on a new parameter of hypergraphs that we define in Subsection \ref{subsec:nesting}, called the \textit{nesting number}. The proof of this result, given in Subsection \ref{subsec:proofoftheoremlowerbound}, depends on a known lower bound for index coding from \cite{neely2013dynamic}. We extend some of our results to the case of PICOD($t$) problems, in which each client demands $t$ messages, for some $t\geq 1$. To the best of our knowledge, our achievability and converse results are novel. We use our results to characterize some PICOD problems for which our converse is tight (Corollary \ref{Corollary-1}). We also give the complete solvability characterization of PICOD problems with maximum degree $\Delta(\ch)\in \{1,2,3\}$, using known results in index coding, in conjunction with our present results (Subsection \ref{subsec:lowdeltaproblems}). The paper concludes in Section \ref{sec:discussion} with some directions for further research. 


\textit{{Some notation and basic definitions:}} We set up some notation and review some basic definitions related to (hyper)graphs that are useful for this work. For some positive integers $m$ and $n$, $n > m$ we denote the set $\{1, 2,....n\}$ by $[n]$ and the sequence $(n, n-1, n-2,....m)$ by $[n : m]$. For sets $A$ and $B$,  $A\backslash B$ denotes the set of elements in $A$ but not in $B$. If $B=\{v\}$, then we also denote $A\setminus B$ as $A\setminus v$. The notation $\lvert A\rvert$ represents the size of set $A$ and  $\emptyset$ denotes an empty set. For a directed graph $G$ with vertex set $\cal V$ and edge set ${\cal E}\subseteq {\cal V}\times {\cal V}$, an induced subgraph of $G$ is a graph with vertex set ${\cal U}\subseteq {\cal V}$ and edge set $\{E\in{\cal E}: E\subseteq {\cal U}\times {\cal U}\}$ (that is, all edges in $G$ between vertices in $\cal U$ are included in the induced subgraph). We now present some definitions related to hypergraphs. A hypergraph $\mc{H}(\mc{V},\mc{E})$ consists of a pair of sets: the \textit{vertex set} denoted by $\mc{V}$, and the \textit{edge set} denoted by $\mc{E}$. The elements of $\mc{E}$ are subsets of the vertex set. We shall also use $\mc{V}(\mc{H})$ and $\mc{E}(\mc{H})$ to refer to the set of vertices $\mc{V}$ and set of edges $\mc{E}$ of a hypergraph $\mc{H}(\mc{V},\mc{E})$, respectively. The \textit{degree} $\D_v$ of a vertex $v$ is defined as the number of edges that contain it (i.e., $\D_v = \lvert \{E \in \mc{E}: v \in E\} \rvert$). The \textit{maximum degree} of hypergraph $\mc{H}(\mc{V},\mc{E})$ is then defined as $\D(\mc{H}) \triangleq \max_{v \in \mc{V}} \D_v$. Two vertices $v_i$ and $v_j$ of $\mc{H}$ are said to be \textit{adjacent} to each other if there is at least one  edge which contains both the vertices. An \textit{independent set} of $\mc{H}$ is a set of vertices such that no two of them in that set are adjacent. For a set  $A \subseteq \mc{V}(\mc{H})$, a set $I\subseteq A$ is called a \textit{maximal independent set in $A$} if $I$ is independent and if for any $s \in A\setminus I$, the set $I \cup \{s\}$ is not independent. For some set $\mc{S} \subseteq \mc{V}(\ch)$, the \textit{induced subgraph} of $\mc{H}$ corresponding to $\mc{S}$, denoted by $\mc{H[\mc{S}]}$, is the hypergraph with vertex set $\mc{S}$, and edge set $\{E \in \mc{E(\mc{H})} : E \subseteq \mc{S}\}$. A \textit{connected component} of a hypergraph $\mc{H}$ is an induced subgraph corresponding to a maximal subset $S$ of vertices such that for each $v \in S$ there exists some $w \in S$ such that $v$ and $w$ are adjacent. The union of two hypergraphs $\mc{H}_1$ and $\mc{H}_2$, denoted by $\mc{H}_1\cup\mc{H}_2$, is the hypergraph with vertex set $\mc{V}(\mc{H}_1)\cup\mc{V}(\mc{H}_2)$ and edge set $\mc{E}(\mc{H}_1)\cup\mc{E}(\mc{H}_2)$. 

\section{System Model for PICOD and hypergraph representation}
\label{sec:systemmodel_and_hypergraph}
We recall the system model of pliable index coding from \cite{brahma2015pliable}. The pliable index coding (PICOD) problem consists of one server with \emph{m} messages from a finite field denoted as $b_i:i\in [m]$, and \emph{n} clients connected to the server by a noise-free broadcast link. The set of indices of the side-information symbols at client $i$ is given by $S_i \subseteq [m]$. Let $R_i = [m]\mathbin{\setminus} S_i$ be referred to as the \textit{request-set} of client $i$. The client $i$ then requires one message from the set of messages $\{b_j : j \in R_i\}$. A pliable index code (or a PICOD scheme) for this problem consists of (a) an encoding scheme at the server which encodes the messages $b_j:j\in[m]$ into $\ell$ symbols from $\mathbb{F}_q$ which the server then broadcasts, and (b) decoding functions at the clients which enable each of them to decode one message with index from its request-set, using the encoded symbols received and its side-information. Here, the quantity $\ell$ is called as the \textit{length} of the PICOD scheme. 
 For some encoding scheme, if a client is able to decode one message in its request-set, it is said to be \textit{satisfied}, and if not, it is said to be \text{unsatisfied}, by that encoding scheme. The goal of pliable index coding is to design PICOD schemes that have length $\ell$ as small as possible, while satisfying all clients.  Observe that we can identify client $i$ with its request-set $R_i$, as any two clients with the same request-set are essentially  identical from the PICOD perspective. The above descriptions and notations extend to the PICOD($t$) scenario also, where each client demands $t$ messages, rather than a single message. In this case, we assume that $|R_i|\geq t, \forall i\in[n]$.

We now briefly recollect the hypergraph model for PICOD as given in \cite{krishnan2021pliable}. A given PICOD problem can be completely captured by a hypergraph $\mc{H}$ with the vertex set $\mc{V} = [m]$ and edge set $\mc{E} = \{R_i \subseteq [m]: i \in [n]\}$. Naturally, every hypergraph also defines a corresponding PICOD problem. We thus do not distinguish between hypergraphs and PICOD problems in this work. We denote by $\beta_q(\mc{H})$ the smallest length of any PICOD scheme for $\mc{H}$ over $\mathbb{F}_q$. Also, let $\beta(\mc{H}) \triangleq \min_{q}\beta_q(\mc{H})$ denote the optimal length of any PICOD scheme for $\mc{H}$ over any finite field. For a PICOD($t$) problem, we denote the optimal PICOD length as $\beta^{(t)}(\ch)$.

\section{A Hypergraph-approach to constructing PICOD schemes in Polynomial-time}
\label{sec:poly-time-delta-algorithm}
In this section, we present Theorem \ref{thm:upperboundDelta} which is a new upper bound (given by $\Delta(\ch)$) on the optimal length of any PICOD scheme for $\ch$. Based on the technique given by the proof of this theorem, Algorithm \ref{Alg1} constructs a PICOD scheme for a given hypergraph $\mc{H}$ using a careful coloring of the vertices of $\mc{H}$. Algorithm \ref{Alg1} guarantees that PICOD problem on $\mc{H}$ can be solved using at most $\Delta({\mc{H}})$ transmissions. Towards proving Theorem \ref{thm:upperboundDelta} (in Subsection \ref{subsec:proofoftheoremlowerbound}) and presenting Algorithm \ref{Alg1} (in Subsection \ref{subsec:algorithm}), we prove Proposition \ref{prop-1} and a couple of helpful lemmas, which provide us insights on how the properties of the hypergraph $\mc{H}$ influence the achievability of a code for $\mc{H}$. Then, we present Algorithm \ref{Alg1} in Subsection \ref{subsec:algorithm} along with its correctness, and show that it has time complexity polynomial in the system parameters. We also provide some comparisons of our algorithm with existing deterministic algorithms in literature in Subsection \ref{subsec:algorithm}.  

\subsection{Some structural results concerning $\beta_q(\ch)$ and a simple upper bound}
We first observe that any hypergraph can be decomposed into its connected components. In other words, for any hypergraph $\mc{H}$, there exists connected components $\mc{H}_i:i=1,\hdots,t$ of $\mc{H}$ such that their vertex sets are mutually disjoint (and thus, so are their edge sets) and $\mc{H}=\cup_{i=1}^t\mc{H}_i$. We now proceed to give our first result which uses this observation. 
\begin{proposition}
\label{prop-1}
Consider a PICOD problem $\mc{H}$ with connected components $\mc{H}_1$, $\mc{H}_2$,....$\mc{H}_t$. Then, for any prime power $q$,
\begin{equation}
\label{eqn:eq-1}
    \beta_q(\mc{H}) =  \max_{i \in [t]}\beta_q(\mc{H}_i)
\end{equation}
\end{proposition}
\renewcommand\qedsymbol{$\blacksquare$}
\begin{proof}
Clearly, we have that $\beta(\mc{H})\geq \beta(\mc{H}_i)$ for any $i\in[t]$, as the vertex sets of $\mc{H}_i:i\in[t]$ are disjoint and since any PICOD scheme for $\mc{H}$ must also satisfy the clients in $\mc{H}_i$. We now prove that $\beta(\mc{H}_i)\leq \max_{i \in [t]}\beta_q(\mc{H}_i)$.

Let $(x_{1,1}, x_{1,2},\hdots, x_{1,\beta_q(\mc{H}_i)})\in \mathbb{F}_q^{\beta_q(\mc{H}_i)}:i\in[t]$ be the array of transmissions made by the server using the optimal PICOD scheme for each connected components $\mc{H}_i:i\in[t]$, respectively. Now, we pad these arrays with zeros (from ${\mathbb F}_q$) to make all of them the same length $\max_{i\in[t]}\beta_q(\mc{H}_i)$. Let the resulting zero-padded arrays be referred to as $X_i:i\in[t]$. Derive a single new array of transmissions $X = \sum_{i \in [t]} X_i$. These $\max_{i\in[t]}\beta_q(\mc{H}_i)$ transmissions comprising $X$ will clearly satisfy all the clients in $\mc{H}$ because (a) each $X_i$ satisfies clients in connected component $\mc{H}_i$, and (b) the vertices (messages) in $\mc{H}_i:i\in[t]$ are disjoint (i.e., any client in $\mc{H}_i$ has all the messages in $\mc{H}_j:j\neq i$ as side-information). Thus, we see that $\beta(\mc{H}_i)\leq \max_{i \in [t]}\beta_q(\mc{H}_i)$. This concludes the proof. 
\end{proof}
Proposition \ref{prop-1} enables us to focus on connected components of the PICOD hypergraph when presenting achievability or converse results. We now present two lemmas that bring out some relationships between the hypergraph topology and the achievability of PICOD schemes for a PICOD hypergraph $\mc{H}$.
The following simple lemma captures the essential characteristic of a `greedy' approach to constructing a PICOD scheme for $\ch$. This will be used to show our Theorem \ref{thm:upperboundDelta}. 
\begin{lemma}
\label{Lemma 2}
For an hypergraph $\ch$ with vertex set $\mc{V}$, let $I$ be an independent set of $\ch$. Then, $\beta_q(\mc{H}) \leq 1 + \beta_q(\mc{H}[\mc{V}\setminus I])$.
\end{lemma}
\renewcommand\qedsymbol{$\blacksquare$}
\begin{proof}
For the independent set $I$, let the  server make one broadcast transmission of the symbol $x$ derived as follows,
\begin{equation}
\label{eq:eq-2}
x = \sum_{i\in I} b_i.
\end{equation}
This transmission satisfies all the clients indexed by $j \in [n]$ such that $R_j \cap I \neq \emptyset$, since $I$ is an independent set, which means $\lvert R_j \cap I \rvert = 1$ if $R_j \cap I \neq \emptyset$. Thus the only unsatisfied clients are the ones indexed by $\{j : R_j \cap I = \emptyset\} = \{j \in [n] : R_j \subseteq \mc{V}\setminus I\}$. These $\{R_j:\forall j ~\text{such that}~ R_j\cap I=\emptyset\}$ are precisely the edges of $\mc{H} [\mc{V}\setminus I]$. Using the optimal code for $\mc{H} [\mc{V}\setminus I]$, we get our result.
\end{proof}
The following observation will be used later in this work. 
\begin{observation}
\label{obs-1}
We observe that the degree of the vertices in $\mc{V}\setminus I$ that are adjacent to atleast one vertex in the set $I$ will decrease by atleast one in the induced hypergraph $\mc{H} [\mc{V}\setminus I]$.
\end{observation}
We now have the following lemma, which we will use in the proof of Theorem \ref{thm:upperboundDelta}. 
\begin{lemma}
\label{Lemma 1}
If $\D(\mc{H}) = 1$, then $\beta_q(\mc{H}) = 1$.
\end{lemma}
\renewcommand\qedsymbol{$\blacksquare$}
\begin{proof}
\textit{Converse}: Observe that for $\D(\mc{H}) = 1$, there is at least one non-empty $R_j$, $j \in [n]$. Thus, at least one transmission is required to satisfy this client, which means $\beta_q(\ch) \geq 1$. 

  \textit{Achievability}: If a hypergraph $\mc{H}$ has $\D(\mc{H}) = 1$, it implies that its vertices either have a degree equal to zero (i.e., not part of request-set of any client) or equal to one. From Proposition \ref{prop-1}, we may assume without loss of generality that  only vertices with degree one are in $\mc{H}$, as this does not affect the calculation of $\beta_q(\mc{H})$. The edges $\mc{E}(\mc{H})$ in such a hypergraph are such that for any two edges $R_i$, $R_j$,  $i, j \in [n]$ and $i \neq j$, $R_i \cap R_j = \emptyset$. Thus, we can find a maximal independent set $I = \{v_1, v_2,...v_n\}$, where $v_i \in R_i$ for $i\in [n]$. Now, from the definition of maximal independent set, the server can make one broadcast transmission given by $\sum_{i \in I} b_{v_i}$. Each client $i$ can decode the message $b_{v_i}\in R_i$ from the above transmission, as all other messages in the sum are in its side-information set. Thus, $\beta_q(\ch)\leq 1$. This completes the proof.
\end{proof}
We now give our upper bound on $\beta_q(\ch)$, which is one of the main results in this work. 
\begin{theorem}
\label{thm:upperboundDelta}
 $\beta_q(\mc{H}) \leq  \D(\mc{H})$.
\end{theorem}
\renewcommand\qedsymbol{$\blacksquare$}
\begin{proof}
Let $\mc{V}(\ch)=\mc{V}$. We prove this by induction on $\Delta(\ch)$. For $\D(\mc{H}) =1$, the claim holds by Lemma \ref{Lemma 1}. This is our base case. Now, for any integer $\D > 1$, assume that for any hypergraph $\ch$ with maximum degree $\D - 1$, it holds that $\beta(\ch)\leq \D-1$. We will now show that the claim holds for any $\mc{H}$ with  $\D(\ch)=\D$. 

Denote $\mc{S}=\{v : \D_v = \D\}$. We take a maximal independent set of this set $\mc{S}$, and call it $I_\D$. Now, we partition the vertices $\mc{V}\setminus I_\D$ into two groups, $\mc{S} \setminus I_\D$ and $\mc{V}\setminus \mc{S}$. Observe that the following are true in $\ch[\mc{V}\setminus I_\D]$.
\begin{itemize}
    \item We claim that the degree of any vertex in $\mc{S} \setminus I_\D$ is at most $\D-1$ in $\ch[\mc{V}\setminus I_\D]$. To see this, we note the following. Firstly, any vertex in $\mc{S} \setminus I_\D$ is adjacent to at least one vertex in $I_\D$ (by the maximality of $I_\D$). By Observation \ref{obs-1}, the degree of this vertex will reduce by at least one from $\D$, in $\ch[\mc{V}\setminus I_\D]$. Thus, the claim follows. 
    \item The degree of all vertices in $\mc{V}\setminus \mc{S}$ is at most $\D-1$, as it was so in $\ch$ itself, by choice of $\mc{S}$. 
\end{itemize}
Therefore we can say that the maximum degree of the hypergraph $\ch[\mc{V}\setminus I_\D]$ is at most $\D-1$ after this transmission. According to the induction assumption, we have that $\beta_q(\ch[\mc{V}\setminus I_\D])\leq \D-1$.  As $I_\D$ is by definition an independent set of $\ch$, we can apply Lemma \ref{Lemma 2} to this. This gives
$\beta_q(\mc{H}) \leq 1 + \beta_q(\mc{H} [\mc{V}\setminus I_\D])\leq 1+(\D-1)=\D$. This completes the proof. 
\end{proof}
\subsection{A polynomial complexity scheme that achieves the length bound in Theorem \ref{thm:upperboundDelta}}
\label{subsec:algorithm} In this subsection, we present Algorithm \ref{Alg1}, which takes a PICOD hypergraph $\mc{H}$ with the largest degree $\D(\mc{H})$ as input and constructs the transmissions of a PICOD scheme of length at most $\D(\ch)$, using a vertex-coloring approach. We also discuss the time-complexity of this algorithm, which turns out to be polynomial in the system parameters. Further, we compare our algorithm with existing deterministic algorithms available in literature. 

Now, we discuss the flow of the algorithm and check its correctness. Algorithm \ref{Alg1} works on a temporary hypergraph $\mc{\ch'}$, initialized with the given input PICOD hypergraph $\mc{\ch}$ (line 1). The algorithm considers (in line 1) a set of $\Delta(\ch)$ colors denoted by $\{\text{color-}1,\hdots,\text{color-}\D(\ch)\}$. In line 1, the variable $c$ initialized to $1$ refers to the index of the color. The loop from lines 4-18 of Algorithm \ref{Alg1} is devoted to finding a maximal independent set (in a greedy fashion) in the set of vertices with degree $\delta$, where $\delta$ is decremented from $\D$ to $1$ (line 2). The existence of such an independent set is kept track of by the variable $Empty$, which is set as $TRUE$ initially (in line 3) and updated to $FALSE$ (in line 7) if at least one vertex with degree $\delta$ is found (and thus included in the independent set) in $\mc{H}'$. As the independent set is being constructed one vertex at a time, the color with index $c$ is assigned to these vertices (line 6). This assignment function is denoted by $C(v)$. The graph $\ch'$ is updated (in lines 8-16) in the following way after adding any vertex $v$ to the independent set (line 6).  
\begin{itemize}
\item  $v$ is removed from the vertex set $\mc{V(\mc{H^\prime})}$ of  $\mc{H^\prime}$ (line 8).
\item Degrees of all other vertices present in any edge containing $v$ are then decreased by one (lines 9-13). 
\item All the edges of $\mc{H^\prime}$ that include vertex $v$ are removed from $\mc{E}(\ch')$ (lines 9,14) (as these will be satisfied by the transmission in line 20 as argued below). 
\end{itemize}
Because of the steps given above, when another vertex (following line 4) is picked to be given color-$c$, it will not be adjacent with any vertex $v$ picked before (as all edges containing $v$ would have been removed). Thus, these steps construct an independent set of vertices with degree $\delta$ in $\ch'$. The maximality of this independent set is ensured by the loop starting with line 4, as we run through all possible vertices in $\ch'$. The sum of the messages corresponding to the vertices in the independent set (identified by color $c$) so constructed is then transmitted (Line 19-20). Note that this transmission satisfies any client in whose request-set these message-vertices are present. This is true, since each such client is missing precisely one message in the independent set constructed (indeed, this is exactly why the set of vertices are	
independent). Note that the color index $c$ is updated (line 21) to keep track of the new (potentially non-empty) independent set corresponding to the vertices with the next $\delta$ value. 

The entire loop is iterated, with the updated graph $\ch'$, until $\mc{E}(\ch')$ becomes empty (checked in lines 23-25), which will eventually happen as $\ch$ has finitely many edges. This also means that each edge $R\in\mc{E}(\ch)$ is removed at some point in line 14, which means there is some vertex in $R$ that is picked as part of some independent set constructed by the algorithm. By the transmission in line 20 corresponding to this independent set, the client $R$ is thus satisfied. Thus, every client is satisfied by Algorithm \ref{Alg1}. Finally, the algorithm stops when $\mc{E}(\ch')$ is empty (line 23).  Observe that the final value of $c$ is the number of transmissions. Since $c\leq \Delta(\ch)$ (by lines 2 and 25), Algorithm \ref{Alg1} provides an achievable scheme of length at most $\D(\ch)$, thus giving a constructive proof for Theorem \ref{thm:upperboundDelta}. 
\begin{algorithm}
    \DontPrintSemicolon
    \KwInput{PICOD Hypergraph $\mc{H} = (\mc{V}, \mc{E})$ with maximum degree $\D(\mc{H})$.}
    \textbf{Initialize:} $\mc{H}^\prime=\mc{H}, c=1$, Colors $\{\text{color-}1,\hdots,\text{color-}\D(\ch)\}.$\;
    \For{$\de=\D(\ch):1$}{$Empty = TRUE$\;
        \For{$v\in\mc{V}(\ch')$}{
            \If{$\D_{v} = \de$}{
            $C(v)=$color-$c$. \;
            $Empty=FALSE$\;
                  $\mc{V(\mc{H^\prime})} \gets \mc{V(\mc{H^\prime})} \setminus v$.\;
                  \For{$R\in \mc{E}(\ch')$}{
                    \If{$v\in R$}{
                    \For{$j \in R\setminus v$}{
                     $\D_j=\D_j -1$.\;
                        }
                    $\mc{E(\mc{H^\prime})} \gets \mc{E(\mc{H^\prime})} \setminus R$.\;      
                    }
                }
              
            }
        }
        \If{$Empty = FALSE$}
        {
        Server transmits $\sum\limits_{\stackrel{v\in\mc{V}(\ch'):}{C(v)=\text{color-}c}}\hspace{-0.4cm}b_{v}$\;
        $c=c+1$.\;
        }
        
        \If{$|\mc{E}(\ch')| = 0$}
        {$break$.\;}
    }
    \caption{Construction of a PICOD scheme for PICOD hypergraph $\ch$ with maximum degree $\D(\mc{H})$} \label{Alg1}
    \label{Algorithm 1}
    \end{algorithm}
\begin{remark} Algorithm \ref{Alg1} can also be understood as constructing a PICOD scheme for $\ch$ via a conflict-free coloring for $\ch$. We refer the reader to \cite{krishnan2021pliable} for more details on this connection. 
\end{remark}

\subsubsection*{Complexity of Algorithm \ref{Alg1}}
For a hypergraph $\ch$ with $m$ vertices and $n$ edges, it can be observed from the design of our algorithm that it runs in polynomial time $\mc{O}(\D(\ch) m^2 n)$. This is because of the following reasons. There are at most $\D(\ch)$ values $\de$ can take, so $\mc{O}(\D(\ch))$ iterations are needed for the outermost \textbf{for} loop. Similarly, the second \textbf{for} loop iterates for $\mc{O}(m)$ rounds. The third \textbf{for} loop iterates for at most $\mc{O}(n)$ rounds and for the innermost \textbf{for} loop the size of each edge can at most be $m$, so we need $\mc{O}(m)$ rounds.
\subsubsection*{Comparison of Algorithm \ref{Alg1} with known deterministic algorithms for PICOD}
In \cite{brahma2015pliable} (see GRCOV algorithm in Section VII of \cite{brahma2015pliable}), a polynomial-time greedy algorithm for PICOD was presented, which identified an exhaustive collection of subsets of clients, each subset of which can be satisfied using one transmission. Note that this is equivalent to finding independent sets of messages in our PICOD hypergraph. Therefore, our algorithm, which finds independent subsets of messages, is similar to the GRCOV algorithm in this sense. However, there is the important distinction that Algorithm \ref{Alg1} searches for the independent sets among vertices \textit{in the descending order of their degree}, whereas GRCOV keeps searching for new vertices to enlarge the independent set, independent of the degree. This key difference enables us to bound the number of transmissions using the maximum degree $\D(\ch)$, whereas the performance of the GRCOV algorithm is not calculated in \cite{brahma2015pliable} analytically. 

In \cite{PolyTime_PICOD}, a polynomial-time algorithm for PICOD is presented (Section IV of \cite{PolyTime_PICOD}), which takes an iterative approach. In each iteration, the messages (remaining) in the system are partitioned into groups based on their effective degrees (which are closely related to their degrees), and then greedily assigned coding vectors such that as many clients as possible are satisfied at each such assignment. It is shown in \cite{PolyTime_PICOD} that, this algorithm needs at most $c\log^2n$ transmissions to satisfy all clients, for some constant $c>1$ (Theorem 1 of \cite{PolyTime_PICOD}). From a more careful analysis of the algorithm, it is fairly easy to see that this worst-case guarantee can be brought down to $c\log(\D(\ch))\log n$ transmissions. An exact comparison with this algorithm seems difficult; however, the worst-case guarantee $c\log(\D(\ch))\log n$ of the algorithm in \cite{PolyTime_PICOD}  can be larger than $\D(\ch)$ (which is the worst-case guarantee of Algorithm \ref{Alg1} in the present work), as it is not difficult to present hypergraphs for which $n$ (the number of edges) is large, whereas $\D(\ch)$ is small. In this sense, our Algorithm provides another facet to achievable PICOD schemes which is different from existing deterministic algorithms.
\section{A new lower bound on $\beta(\ch)$ based on nested hyperedges}
\label{sec:lowerbound_via_nesting}
Having discussed achievability of PICOD, we now turn our attention to converse arguments. In this section, we identify a substructure of the PICOD hypergraph called a \textit{nested collection} of hyperedges and obtain a lower bound for PICOD based on this substructure. Using this lower bound, we also present some classes of problems for which $\beta(\ch)=\D(\ch)$. Using our present results and based on existing results in index coding literature, we characterize the optimal lengths of all PICOD problems with $\D(\ch)\in\{1,2,3\}$. 
\subsection{A nested collection of hyperedges of $\ch$} 
\label{subsec:nesting}
\begin{definition}
\label{defn:nestingcollection}
Consider a hypergraph $\mc{H}(\mc{V}, \mc{E})$. A collection of $L$ subsets of hyperedges of $\ch$ given by $\{ \mc{E}_i\subseteq \mc{E}: i\in[L]\}$ is said to be a \textit{nested collection} of hyperedges of $\ch$ if it satisfies the following properties 
\begin{itemize}
    \item The set of edges $\mc{E}_i:i\in[L]$ has cardinality $\lvert \mc{E}_i \rvert = 2^{i-1}$, respectively. 
    \item For any $R\in \mc{E}_i:i\in[L-1]$ there exists a pair of non-empty edges $R'$ and $R''$ in $\mc{E}_{i+1}$ such that $R', R'' \subset R$ and $R'\cap R''=\emptyset$.
\end{itemize}
 We refer to the set $\mc{E}_i$ as the \textit{level-$i$} edges of the nested collection. We call the number of levels $L$ in such a collection as its \textit{nesting length}. We define the maximum nesting length of any nested collection in a hypergraph $\mc{H}$ as the \textit{nesting number of $\ch$}, and denote it by $\eta({\mc{H}})$.
\end{definition}
In this section, we will prove the following result. 
\begin{theorem}
\label{thm:lowerboundeta}
  Let $\mc{H}$ be a PICOD hypergraph. Then, $\beta(\mc{H}) \geq \eta(\mc{H})$.
\end{theorem}
The proof of Theorem \ref{thm:lowerboundeta} is based on the maximum acyclic induced subgraph bound for index coding, obtained in \cite{neely2013dynamic}. Before presenting the proof in Subsection \ref{subsec:proofoftheoremlowerbound}, we present two examples to illustrate the quantity $\eta(\ch)$ as well as the lower bound in Theorem \ref{thm:lowerboundeta}.
\begin{example}
\label{example1} 
Consider a PICOD hypergraph with vertices indexed by $[12]$, and $12$ edges given by $R_1 = \{1, 2, 3, 4, 5, 6, 8, 10, 12\}$, $R_2 = \{1, 2, 4, 8\}$, $R_3 = \{3, 5, 6, 10\}$, $R_4 = \{1, 2, 8\}$, $R_5 = \{4\}$, $R_6 = \{5, 10\}$, $R_7 = \{3\}$, $R_8 = \{1, 7, 9\}$, $R_9 = \{3, 9\}$, $R_{10} = \{11\}$, $R_{11} = \{1, 11, 12\}$, $R_{12} = \{2, 11\}$.  Fig. \ref{fig:fig 1} shows a subgraph of $\ch$ consisting of a nested collection of hyperedges with nesting length  $L=3$.  
Here, the nested collection of hyperedges is $\{ \mc{E}_1, \mc{E}_2, \mc{E}_3\}$, where $\mc{E}_1 = \{R_1\}$, $\mc{E}_2 = \{R_2, R_3\}$ and $\mc{E}_3 = \{R_4, R_5, R_6, R_7\}$. It is easy to verify that the conditions required in Definition \ref{defn:nestingcollection} are met by this collection. In fact, it can be shown that there is no nested collection of length $4$ for this example. Hence, $\eta(\ch)=3$. 
\begin{figure}[htbp]
    \centering
\includegraphics[width=0.35\textwidth]{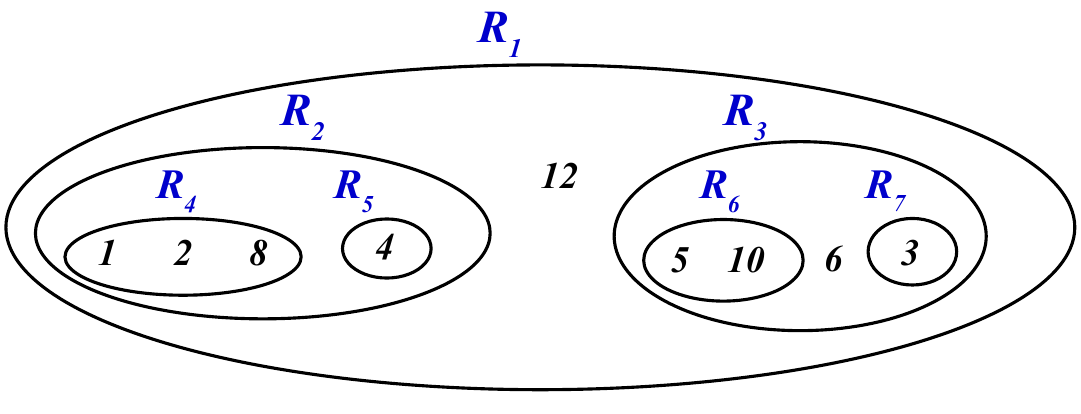}
    \caption{Nested collection of hypergraph $\mc{H}$ in Example \ref{example1} with nesting length $\eta(\mc{H})=3$. Each ellipse represents an edge, denoted by $R_j:j\in[7]$. The vertices are represented by integers.}
    \label{fig:fig 1}
\end{figure}
Now, we describe the intuition behind why $\beta(\ch)\geq 3$. Suppose, in some PICOD scheme which satisfies all the clients, client $R_1\in\mc{E}_1$ is satisfied by message $b_{i_1}$ (where $i_1\in R_1$). Observe that by the structure of the nested collection, there is some $R_{j_2}\in\mc{E}_2$ for which $i_1\in ([m]\setminus R_{j_2})$ (i.e., $b_{i_1}$ in its side-information). Now assume that this client $R_{j_2}$ is satisfied by message $b_{i_2}$ (where $i_2\in R_{j_2}$). Note that $i_1,i_2$ must be distinct. Once again, there will be some client $R_{j_3}\in\mc{E}_3$ which is having both $i_1,i_2 \in [m]\setminus R_{j_3}$ (i.e., $b_{i_1},b_{i_2}$ are both in side-information set of client $j_3$), by the structure of the nested collection. This means that $R_{j_3}$ is satisfied only by another distinct message $b_{i_3}$ ($i_3\in R_{j_3}$). Again, because of the nested structure, we observe that the $R_{j_3}$ includes in its side-information all the side-information symbols of $R_{j_2}$ (and thus of $R_1$ as well). Because of this reason,  $R_{j_3}$ can decode all the three symbols $b_{i_1},b_{i_2},$ and $b_{i_3}$. A simple information theoretic argument would then imply that at least $3$ transmissions are required in this code. Thus, $\beta(\ch)\geq 3$.

Finally, we see that there is a PICOD scheme of length $3$ for $\ch$. The following three transmissions serve as a PICOD scheme: $x_1 = b_{1} + b_3 + b_{4} + b_{5}$, $x_2 = b_{2} + b_{6}$ and $x_3 = b_4 + b_{11}$. It can be verified that the first transmission $x_1$ satisfies the clients $R_4,R_5,R_6,R_7,R_8,R_9$; $x_2$ satisfies clients $R_2, R_3$; and $x_3$ satisfies $R_1,R_{10},R_{11}$,and $R_{12}$. Thus, for this example, we have $\beta(\ch)=\eta(\ch)=3$.
\end{example}
\begin{example}
\label{example2}
We also give another example for a hypergraph for which $\eta(\ch) = 2$, shown in Fig. \ref{fig:fig 2}. For this example, vertex set $\mc{V} = [9]$ and edge set $\mc{E} = \{R_i:i\in[9]\}$, where $R_i = \{i\}:i \in [1:6]$, $R_7 = \{
1, 2, 7, 8\}$, $R_8 = \{
3, 4, 7, 9\}$, and $R_9 = \{
5, 6, 8, 9\}$. It is easy to see that there is a nested collection of length $L=2$ given by $\{ \mc{E}_1, \mc{E}_2\}$, where $\mc{E}_1 = \{R_7\}$ and $\mc{E}_2 = \{R_1, R_2\}$. It can be seen that $\eta(\mc{H}) = 2$. By Theorem \ref{thm:lowerboundeta} (via similar intuitive arguments as in Example \ref{example1}), the number of transmissions in a PICOD scheme is at least two. Further, as $\D(\ch)=2$, it must be that $\beta(\ch)=2$. One such optimal scheme consists of the two transmissions: $x_1 = b_{1} + b_{3} + b_{5}$ and $x_2 = b_{2} + b_{4} + b_{6}$. Note that $x_1$ satisfies clients $R_1, R_3, R_5, R_7, R_8$ and $R_9$; $x_2$ satisfies the remaining clients $R_2, R_4$ and $R_6$. Thus for this example, $\beta(\ch)=\eta(\ch)=\D(\ch)=2$.
\begin{figure}[htbp]
    \centering
    \includegraphics[width=0.25\textwidth]{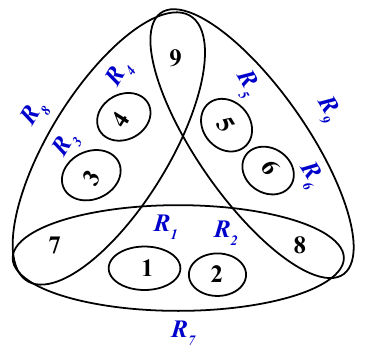}
    \caption{Hypergraph corresponding to Example \ref{example2} with nesting length $\eta(\mc{H})=2$. Integers represent messages, ellipses represent edges $R_j:j\in[9]$.}
    \label{fig:fig 2}
\end{figure}
\end{example} 

\subsection{Proving Theorem \ref{thm:lowerboundeta}}
\label{subsec:proofoftheoremlowerbound}
To prove Theorem \ref{thm:lowerboundeta}, we first recall the formulation of the index coding problem from \cite{ISCOD_BiK,BiK}, the acyclic induced subgraph lower bound for index coding from \cite{neely2013dynamic}, and their relationship to the PICOD problem. The system model of an index coding problem, denoted by ${\cal I}$, consists of a server containing a set of $m'$ messages (symbols $b_i:i\in[m']$ from a finite field) connected via a broadcast link to $n'$ clients. As in the PICOD problem, the client $i$ possesses side-information symbols, indexed by $S_i'\subset[m']$. Let $R_i'=[m']\setminus S_i'$ denote the non-side-information symbols at client $i$. In contrast to the PICOD problem, in the index coding problem, a specific message (say, with index $d_i\in R_i'$) is demanded by the client $i$. An index code of length $l$ is a set of $l$ transmissions made by the server to satisfy the client demands. We denote by $\beta({\cal I})$ the smallest length of any index code (over any finite field) for $\cal I$. In a PICOD problem $\ch$, if the decoding choices (messages demanded) of the PICOD clients are fixed, then it reduces to an index coding problem. Let the fixed set of decoding choices for clients be denoted by $D=(b_{d_1},\hdots,b_{d_n})$, where $b_{d_i}$ denotes the decoding choice of client $i$. Let the index coding problem obtained from the PICOD hypergraph $\ch$ corresponding to the decoding choices $D$ be denoted by $\ch_D$. It then holds from Lemma 1, \cite{ong2019optimal} that \begin{align}\label{eqn:lowerboundmin}
\beta(\ch)=\min_D\beta(\mc{H}_D),
\end{align} 
where the minimization is over all possible decoding choices $D$. We now recall the lower bound for index coding from \cite{neely2013dynamic}, which applies to $\beta(\ch_D)$. Then, using (\ref{eqn:lowerboundmin}), we can prove the lower bound in Theorem \ref{thm:lowerboundeta}. 


For a given index coding problem ${\cal I}$, we can construct a directed bipartite graph $G({\cal I})$ with left-vertices $R_j':j\in[n']$ and right-vertices $b_j:j\in[m']$, as in \cite{neely2013dynamic}. The definition of the edges in this graph is as follows. 
\begin{itemize}
    \item For each client $i$, a single solid arrow points from the demanded message $b_{d_i}$ to $R_i'$.
    \item For each client $i$, dashed arrows point from $R_i'$ to the messages that are present in its side-information set $\{b_j:j\in S_i'\}$.
\end{itemize}
A lower bound for index coding problem was then derived in \cite{neely2013dynamic} by extracting an acyclic induced subgraph from $G({\cal I})$ by performing some pruning operations. Further, it is not difficult to show that these pruning operations can extract any acyclic induced subgraph of $G({\cal I})$. Thus, from \cite{neely2013dynamic} (Theorem 1 and Lemma 1 in \cite{neely2013dynamic}), if $L$ is the number of message-indices (right-vertices) in any acyclic induced subgraph of $G({\cal I})$, then 
\begin{align}
\label{inequality:IC_lowerbound_acyclic}
    \beta({\cal I})\geq L.
\end{align}

We are now ready to show the proof of Theorem \ref{thm:lowerboundeta}. 
\begin{proof}[Proof of Theorem \ref{thm:lowerboundeta}]

Let there be a nested collection of hyperedges in $\ch$ with $L$ subsets of $\mc{E}(\ch)$, denoted by $\{\mc{E}_1,\hdots,\mc{E}_L\}$. Let $D=(b_{d_1},\hdots,b_{d_n})$ be some collection of decoding choices for the $n$ clients of the PICOD problem represented by $\ch$.  We shall show the following claim:

\textit{Claim:} For any choice of $D$, the bipartite graph representation of the index coding problem $\ch_D$ contains an acyclic induced subgraph with $L$ messages.

If the above claim is true, using (\ref{inequality:IC_lowerbound_acyclic}), we would have that $\beta(\ch_D)\geq L$. As $D$ is arbitrary in the claim, and by (\ref{eqn:lowerboundmin}), we would then have $\beta(\ch)=\min_D\beta(\ch_D)\geq L$. Finally, we can choose $L=\eta(\ch)$, by definition of the nesting number $\eta(\ch)$. This would complete the proof. 

We now prove the claim. Without loss of generality, let $\mc{E}_1=\{R_1\}$. Recall that $b_{d_1}$ is the decoding choice of $R_1$. Then, by the nesting structure, there is some client in $\mc{E}_2$, say $R_2\subset R_1$, such that $d_1\in [m]\setminus R_2$. Thus, the decoding choice $b_{d_2}$ of this client $R_2$ must be such that $d_1,d_2$ are distinct. Further, as $R_2\subset R_1$, we must have that $d_2\in R_1$ as well. By a similar argument, for each $i\in[L]$, there is some client $R_i\in \mc{E}_i$, such that the decoding choices of the prior clients $R_j:j\in[i-1]$ are in its side-information set, i.e., $\{d_j:j\in[i-1]\}\subseteq [m]\setminus R_i$. Thus, the decoding choice $b_{d_i}$ of $R_i$ must be distinct from $\{b_{d_j}:j\in[i-1]\}$. In other words, we must have clients $R_i:i\in[L]$ such that the decoding choices of $R_i:i\in[L]$ satisfy the following succinct condition.
\begin{align}
\label{eqn:acyclic_condition}d_i \in \bigcap\limits^i_{j=1}R_j
\setminus \bigcup\limits^L_{k=i+1}R_k, \forall i \in [L]
\end{align}
Note that this also means $|\{b_{d_i}:i\in[L]\}|=L$.

By (\ref{eqn:acyclic_condition}), we see that the bipartite graph representation of the index coding problem $\ch_D$ contains an induced subgraph with the $L$ messages $\{b_{d_i}:i\in[L]\}$ as represented in Fig. \ref{fig:bipartite_acyclic}. It is easy to prove that this is acyclic because:
\begin{itemize}
    \item If $b_{d_i}$ points to $R_j$, then $i=j$.
    \item If $R_j$ points to $b_{d_i}$, then $i < j$.
\end{itemize}
Thus, we never have a `downward' pointing arrow in the subgraph of Fig. \ref{fig:bipartite_acyclic}, hence there are no cycles. This completes the proof of the claim, and hence the proof of Theorem \ref{thm:lowerboundeta}. 

\begin{figure}[htbp]
    \centering
    \includegraphics[width=0.3\textwidth]{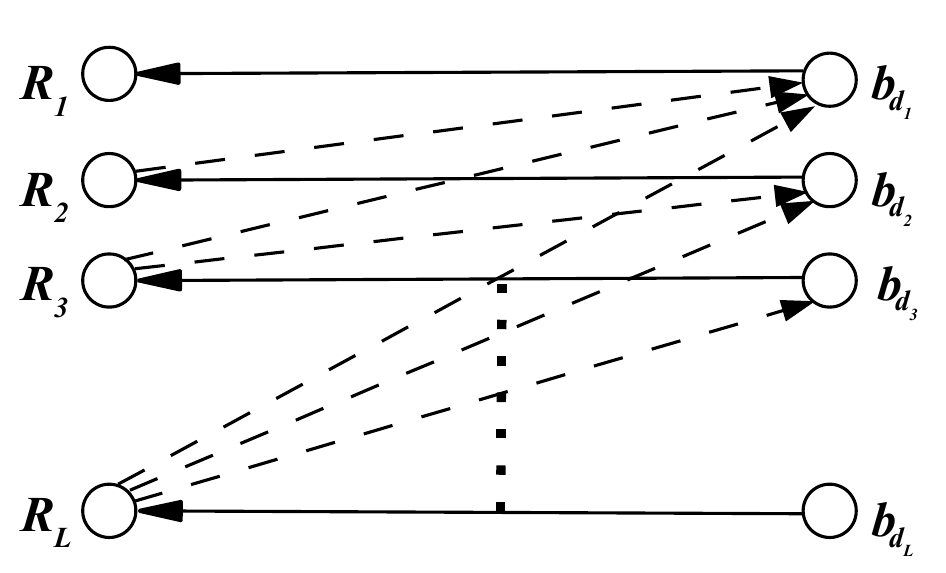}
    \caption{Acyclic induced subgraph used in proof of Theorem \ref{thm:lowerboundeta}: Solid edges from message vertex $b_{d_i}$ to client $R_i$ indicates that $b_{d_i}$ is the decoding choice for client $i$, and dotted edges from $R_i$ to $b_{d_j}$  indicates that $b_{d_j}$ is in side-information set of client $R_i$.}
    \label{fig:bipartite_acyclic}
\end{figure}
\end{proof}
\begin{remark}
Theorem \ref{thm:lowerboundeta} shows that the existence of a special substructure (the nested collection) leads to a lower bound for PICOD. We use a decoding chain argument as in prior works on converses for PICOD \cite{liu2019tight,ong2019optimal,ong2019improved}, which were essentially based on the index coding converse from \cite{neely2013dynamic}. While the work \cite{liu2019tight} mainly used this for the case of specially structured PICOD problems, the works \cite{ong2019optimal,ong2019improved} focussed on identifying structural characteristics in the set of `absent' clients (i.e., subsets of $[m]$ which no client in the given problem has as its request-set). In the present work, we focus on the `present' clients, i.e. those request-sets which are included in the given problem. This gives a different, and possibly a more natural method to characterize a converse for PICOD. 
\end{remark}
The following corollary gives a class of hypergraphs for which the upper bound from Theorem \ref{thm:upperboundDelta} is tight. The corollary follows directly from Theorems \ref{thm:upperboundDelta} and \ref{thm:lowerboundeta}. 
\begin{corollary}
\label{Corollary-1}
Let $\mc{H}$ be a PICOD hypergraph with $\eta(\mc{H}) = \D(\mc{H})$. Then, $\beta_q(\mc{H}) = \D(\mc{H})$ for every prime power $q$, and hence $\beta(\ch)=\D(\ch)$.
\end{corollary}
We now generalize the bound in Theorem \ref{thm:lowerboundeta} to the PICOD($t$) scenario. Recall that if a hypergraph $\ch(\mc{V}, \mc{E})$ represents a PICOD($t$) problem, then $|R|\geq t$ for each $R\in\mc{E}$. \begin{lemma}
\label{lemma:generalizedlowerboundetaforPICOD(t)}
For a PICOD($t$) hypergraph $\mc{H}(\mc{V}, \mc{E})$, suppose there exists  a collection of subsets of $\mc{E}$ written as $\{ \mc{E}_i\subseteq \mc{E}: i\in[L]\}$, such that for each $i\in[L-1]$, the following condition holds.
\begin{itemize}
    \item For each $R\in \mc{E}_i$, for any subset $T\subset R$ with $|T|=t$, there exists an edge $R'\in\mc{E}_{i+1}$ such that $T\cap R'=\emptyset$.
\end{itemize}
Then $\beta^{(t)}(\ch)\geq tL$.
\end{lemma}
We omit the proof of Lemma \ref{lemma:generalizedlowerboundetaforPICOD(t)}, as it is essentially identical to Theorem \ref{thm:lowerboundeta}. Clearly, choosing the collection of subsets with the largest possible value for $L$ gives us the best bound that can be obtained from Lemma \ref{lemma:generalizedlowerboundetaforPICOD(t)}.

\subsection{Complete categorization of $\D(\mc{H}) = 1, 2, 3$}
\label{subsec:lowdeltaproblems}


In this subsection, we characterize the PICOD problems with $\D(\ch)\in\{1,2,3\}$ completely. These results follow directly from known results in index coding, in conjunction with our results. 

From (\ref{eqn:lowerboundmin}), we see that if there is some decoding choice $D$ such that $\beta(\ch_D)=\ell$, then clearly $\beta(\ch)\leq \ell$. From this, the following lemma is obtained. 
\begin{lemma}
\label{lemma:lengthone}
$\beta(\ch)=1$, if and only if there exists some decoding choices $D=(b_{d_1},\hdots,b_{d_n})$ for the $n$ clients in $\ch$, such that the set $\cup_{i\in[n]}\{d_i\}$ is an independent set of $\ch$.
\end{lemma}
\begin{proof}
The achievability is easy to see, as transmitting a sum of the unique messages in $D$ satisfies all receivers. Now, suppose the condition does not hold. That is, for every decoding choice $D$, there are at least two clients (say $R_1$ and $R_2$) with distinct decoding choices (say $d_1$ and $d_2$) such that one of these clients (say $R_1$) contains both $d_1,d_2$ in its request-set. Thus, the bipartite graph corresponding to the index coding problem $\ch_D$ has a acyclic induced subgraph of size $2$ (consisting of $R_1,R_2,b_{d_1},b_{d_2}$). By Lemma \ref{eqn:acyclic_condition}, we see that $\beta(\ch_D)\geq 2$. Thus, the proof is complete. 
\end{proof}
\begin{remark}
The result in Lemma \ref{lemma:lengthone} was previously proved in \cite{liu2019tight} (Proposition 8 in \cite{liu2019tight}), using similar arguments. Our proof directly relies upon the lower bound for index coding from \cite{neely2013dynamic} and is therefore more succinct.
\end{remark}
Matching necessary and sufficient conditions for length-2 index coding were derived in \cite{Blasiak_BoundingIC,6855359_Maleki_TIT_IC_IA}. We refer the reader to these works for the conditions (see Theorem V.9 in \cite{Blasiak_BoundingIC}, Corollary 1 in \cite{6855359_Maleki_TIT_IC_IA}). Clearly, if there is some decoding choice collection $D$ for the given PICOD problem $\ch$ such that the resultant index coding problem $\ch_D$ is length-$2$ solvable (which can be verified by the conditions in \cite{Blasiak_BoundingIC,6855359_Maleki_TIT_IC_IA}), then the given PICOD problem also has a PICOD scheme of length two. We thus see that we have a characterization of all PICOD problems with $\Delta\in\{1,2,3\}$, by using our results in this subsection and the results from \cite{Blasiak_BoundingIC,6855359_Maleki_TIT_IC_IA}, in conjunction with Theorem \ref{thm:upperboundDelta}. 










\section{Discussion}
\label{sec:discussion}
In this work, we have presented a achievability result	
for PICOD based on the maximum degree of a hypergraph	
which represents the PICOD problem. A polynomial-time	
deterministic algorithm meeting the achievability result was	
also presented, adding to the list of polynomial-complexity	
algorithms for PICOD in literature. We also showed a converse	
based on a special substructure in the hypergraph called a	
nested collection of hyperedges. Recent results also indicate	
that there is a simple algorithm to compute the quantity	
$\eta(\mc{H})$. Naturally, we would be interested to see if other low-	
complexity algorithms can be found for PICOD with achiev-	
able lengths better than the known algorithms. Also, it would	
be interesting to see if we could characterize achievability or	
converses of PICOD problems based on other more refined	
hypergraph parameters.

Also, while the results in Subsection \ref{subsec:lowdeltaproblems} suggest checking for some hypergraph substructures for deciding on the feasibility of length-$1$ and length-$2$ PICOD schemes, it is still not clear whether the existence of such substructures in a given hypergraph can be checked with polynomial complexity. This is an interesting direction for future work. It would also be worthwhile to characterize PICOD problems for larger values of maximum degree $\D(\ch)$ greater than $3$. 
\bibliographystyle{IEEEtran}
\bibliography{IEEEabrv,ArXiv_ResearchGate_2022_PICOD_hyp}

\begin{thebibliography}{10}
\providecommand{\url}[1]{#1}
\csname url@samestyle\endcsname
\providecommand{\newblock}{\relax}
\providecommand{\bibinfo}[2]{#2}
\providecommand{\BIBentrySTDinterwordspacing}{\spaceskip=0pt\relax}
\providecommand{\BIBentryALTinterwordstretchfactor}{4}
\providecommand{\BIBentryALTinterwordspacing}{\spaceskip=\fontdimen2\font plus
\BIBentryALTinterwordstretchfactor\fontdimen3\font minus
  \fontdimen4\font\relax}
\providecommand{\BIBforeignlanguage}[2]{{%
\expandafter\ifx\csname l@#1\endcsname\relax
\typeout{** WARNING: IEEEtran.bst: No hyphenation pattern has been}%
\typeout{** loaded for the language `#1'. Using the pattern for}%
\typeout{** the default language instead.}%
\else
\language=\csname l@#1\endcsname
\fi
#2}}
\providecommand{\BIBdecl}{\relax}
\BIBdecl

\bibitem{brahma2015pliable}
S.~Brahma and C.~Fragouli, ``Pliable index coding,'' \emph{IEEE Transactions on
  Information Theory}, vol.~61, no.~11, pp. 6192--6203, 2015.

\bibitem{ISCOD_BiK}
Y.~Birk and T.~Kol, ``Coding on demand by an informed source (iscod) for
  efficient broadcast of different supplemental data to caching clients,''
  \emph{IEEE Transactions on Information Theory}, vol.~52, no.~6, pp.
  2825--2830, 2006.

\bibitem{song2019pliable_shuffling}
L.~Song, C.~Fragouli, and T.~Zhao, ``A pliable index coding approach to data
  shuffling,'' \emph{IEEE Transactions on Information Theory}, vol.~66, no.~3,
  pp. 1333--1353, 2019.

\bibitem{PolyTime_PICOD}
L.~{Song} and C.~{Fragouli}, ``A polynomial-time algorithm for pliable index
  coding,'' \emph{IEEE Transactions on Information Theory}, vol.~64, no.~2, pp.
  979--999, 2018.

\bibitem{liu2019tight}
T.~Liu and D.~Tuninetti, ``Tight information theoretic converse results for
  some pliable index coding problems,'' \emph{IEEE Transactions on Information
  Theory}, vol.~66, no.~5, pp. 2642--2657, 2019.

\bibitem{8849812_ShanuRaj_CodeConstrPIC}
S.~Sasi and B.~S. Rajan, ``Code construction for pliable index coding,'' in
  \emph{2019 IEEE International Symposium on Information Theory (ISIT)}, 2019,
  pp. 527--531.

\bibitem{ong2019optimal}
L.~Ong, B.~N. Vellambi, and J.~Kliewer, ``Optimal-rate characterisation for
  pliable index coding using absent receivers,'' in \emph{2019 IEEE
  International Symposium on Information Theory (ISIT)}.\hskip 1em plus 0.5em
  minus 0.4em\relax IEEE, 2019, pp. 522--526.

\bibitem{ong2019improved}
L.~Ong, B.~N. Vellambi, J.~Kliewer, and P.~Sadeghi, ``Improved lower bounds for
  pliable index coding using absent receivers,'' \emph{arXiv preprint
  arXiv:1909.11850}, 2019.

\bibitem{krishnan2021pliable}
P.~Krishnan, R.~Mathew, and S.~Kalyanasundaram, ``Pliable index coding via
  conflict-free colorings of hypergraphs,'' in \emph{2021 IEEE International
  Symposium on Information Theory (ISIT)}, 2021, pp. 214--219.

\bibitem{even2003conflict}
G.~Even, Z.~Lotker, D.~Ron, and S.~Smorodinsky, ``Conflict-free colorings of
  simple geometric regions with applications to frequency assignment in
  cellular networks,'' \emph{SIAM Journal on Computing}, vol.~33, no.~1, pp.
  94--136, 2003.

\bibitem{neely2013dynamic}
M.~J. Neely, A.~S. Tehrani, and Z.~Zhang, ``Dynamic index coding for wireless
  broadcast networks,'' \emph{IEEE Transactions on Information Theory},
  vol.~59, no.~11, pp. 7525--7540, 2013.

\bibitem{BiK}
Y.~{Birk} and T.~{Kol}, ``Coding on demand by an informed source (iscod) for
  efficient broadcast of different supplemental data to caching clients,''
  \emph{IEEE Transactions on Information Theory}, vol.~52, no.~6, pp.
  2825--2830, 2006.

\bibitem{Blasiak_BoundingIC}
A.~Blasiak, R.~Kleinberg, and E.~Lubetzky, ``Broadcasting with side
  information: Bounding and approximating the broadcast rate,'' \emph{IEEE
  Transactions on Information Theory}, vol.~59, no.~9, pp. 5811--5823, 2013.

\bibitem{6855359_Maleki_TIT_IC_IA}
H.~Maleki, V.~R. Cadambe, and S.~A. Jafar, ``Index coding—an interference
  alignment perspective,'' \emph{IEEE Transactions on Information Theory},
  vol.~60, no.~9, pp. 5402--5432, 2014.

\end{thebibliography}
\end{document}